\documentclass[final,3p,times,authoryear,11pt]{elsarticle}

\usepackage{amsmath,amssymb,amsthm}
\usepackage{graphicx}
\usepackage{microtype}
\usepackage{url}

\biboptions{authoryear,round}

\journal{Statistics \& Probability Letters}

\newtheorem{proposition}{Proposition}
\newtheorem{theorem}{Theorem}
\newtheorem{corollary}{Corollary}

\newcommand{\HC}{\mathrm{C}^{+}}
\newcommand{\E}{\mathbb{E}}
\newcommand{\Pp}{\mathbb{P}}
\newcommand{\Var}{\operatorname{Var}}
\newcommand{\med}{\operatorname{med}}
\newcommand{\sech}{\operatorname{sech}}

\begin{document}

\begin{frontmatter}


\title{
Pathwise Shrinkage and Probabilistic Calibration of the
Multiplicative Half-Cauchy Process
}

\author[aff1]{Siddhesh Kulkarni\corref{cor1}}
\ead{siddhesh.s.kulkarni@gmail.com}

\cortext[cor1]{Corresponding author}

\address[aff1]{
Bristol Myers Squibb, 5 Giralda Farms, Madison, NJ 07940, USA
}


\begin{abstract}
We characterize the multiplicative half-Cauchy process through its log-scale random-walk representation, establishing shrinkage, expansion, and critical regimes, almost-sure summability, local rebounds, and effective dimension. We also derive a semi-analytic calibration rule linking interpretable prior probability statements to its scale parameter.
\end{abstract}

\begin{keyword}
Bayesian shrinkage \sep
multiplicative half-Cauchy process \sep
ordered shrinkage priors \sep
infinite factor models \sep
prior calibration \sep
effective dimension
\end{keyword}

\end{frontmatter}

\section{Introduction}

Bayesian over-complete factorization models often introduce more
latent components than are expected to be required and use ordered
shrinkage priors to suppress redundant ones. Examples include the
multiplicative gamma process
\citep{bhattacharya2011,durante2017},
cumulative shrinkage \citep{legramanti2020}, and generalized
infinite factorizations \citep{schiavon2022}.
A practical challenge is to understand both the effective dimension
induced by such priors and how their hyperparameters should be
selected.

The multiplicative half-Cauchy process (MHCP) and its use within a Bayesian factor model for sparse canonical correlation analysis were first developed in \citet{Kulkarni2022} and later published in \citet{kulkarni2026}. In its motivating model,
\begin{equation}
a_{jh}^{(m)}
\mid \tau_m,\eta_h,\lambda_{jh}^{(m)}
\sim
N\left(
0,\tau_m^2\eta_h^2\{\lambda_{jh}^{(m)}\}^2
\right).
\label{eq:loading}
\end{equation}
Thus $\eta_h$ is a column-specific standard-deviation multiplier.
Earlier work established full support, increasing concentration near
zero, upper-tail flexibility, and
$\med(\eta_h)=\zeta^{h-1}$.

Here we study the MHCP itself. We derive its pathwise phase behavior,
summability and threshold-exceedance properties, local rebound
probabilities, and an interpretable calibration for $\zeta$.
Throughout, ``increasing shrinkage'' refers to decreasing typical
scale and long-run behavior, not monotone sample paths or
first-order stochastic ordering.

\section{Log-scale representation}

The MHCP of \citet{kulkarni2026} is
\begin{equation}
\widetilde{\eta}_1=1,
\qquad
\widetilde{\eta}_h\sim\HC(0,\zeta),
\quad h\geq2,
\label{eq:increments}
\end{equation}
independently, with
\begin{equation}
\eta_h
=
\prod_{\ell=1}^{h}
\widetilde{\eta}_{\ell}.
\label{eq:mhcp}
\end{equation}

For $h\geq2$, define
\[
U_h=\widetilde{\eta}_h/\zeta,
\qquad
Y_h=\log U_h.
\]

\begin{proposition}[Centered log representation]
\label{prop:log}

The variables $Y_2,Y_3,\ldots$ are i.i.d. with
\begin{equation}
f_Y(y)
=
\frac{1}{\pi\cosh y},
\qquad
\phi_Y(t)
=
\sech\left(\frac{\pi t}{2}\right),
\label{eq:logdensity}
\end{equation}
and
\begin{equation}
\log\eta_h
=
(h-1)\log\zeta
+
\sum_{\ell=2}^{h}Y_\ell,
\label{eq:decomp}
\end{equation}
where
\[
\E(Y_\ell)=0,
\qquad
\Var(Y_\ell)=\frac{\pi^2}{4}.
\]

\end{proposition}

\begin{proof}
For $U\sim\HC(0,1)$, the transformation $Y=\log U$ gives
\[
f_Y(y)
=
\frac{2e^y}{\pi(1+e^{2y})}
=
\frac{1}{\pi\cosh y}.
\]
Its characteristic function is
$\sech(\pi t/2)$
\citep{harkness1968,ding2014three}.
Equation \eqref{eq:decomp} follows from
$\widetilde{\eta}_h=\zeta U_h$.
\end{proof}

Consequently,
\begin{equation}
\E(\log\eta_h)
=
(h-1)\log\zeta,
\qquad
\Var(\log\eta_h)
=
(h-1)\frac{\pi^2}{4}.
\label{eq:logmoments}
\end{equation}

Thus $\zeta$ changes the global log-scale drift without changing the
centered increment law. Symmetry gives
\begin{equation}
\med(\eta_h)=\zeta^{h-1}.
\label{eq:median}
\end{equation}

For $-1<q<1$,
\begin{equation}
E(\eta_h^q)
=
\zeta^{q(h-1)}
\left\{
\sec\left(\frac{\pi q}{2}\right)
\right\}^{h-1},
\label{eq:fractional-moment}
\end{equation}
while $E(\eta_h)=\infty$ for $h\geq2$.

\section{Local and pathwise behavior}

The process need not decrease from one component to the next.

\begin{proposition}[Multiplicative excursions]
\label{prop:excursion}

For $h\geq2$ and $c>0$,
\begin{equation}
\Pp(\eta_h>c\eta_{h-1})
=
\frac{2}{\pi}
\arctan\left(
\frac{\zeta}{c}
\right).
\label{eq:excursion}
\end{equation}

\end{proposition}

\begin{proof}
Since
$\eta_h/\eta_{h-1}=\widetilde{\eta}_h$,
the result follows directly from the half-Cauchy distribution
function.
\end{proof}

For $\zeta=0.7$, one-, two-, and five-fold upward moves have
probabilities approximately $0.389$, $0.214$, and $0.089$,
respectively. Hence substantial local rebounds remain possible even
under global shrinkage.

\begin{theorem}[Pathwise phase behavior]
\label{thm:phase}

As $h\to\infty$,
\begin{equation}
\frac{\log\eta_h}{h-1}
\longrightarrow
\log\zeta,
\qquad
\eta_h^{1/(h-1)}
\longrightarrow
\zeta
\quad\text{a.s.}
\label{eq:phase}
\end{equation}

Hence
\[
0<\zeta<1
\Longrightarrow
\eta_h\to0,
\qquad
\zeta>1
\Longrightarrow
\eta_h\to\infty
\quad\text{a.s.}
\]
At $\zeta=1$,
\begin{equation}
\liminf_{h\to\infty}\eta_h=0,
\qquad
\limsup_{h\to\infty}\eta_h=\infty
\quad\text{a.s.}
\label{eq:critical}
\end{equation}

\end{theorem}

\begin{proof}
Apply the strong law to \eqref{eq:decomp}. At $\zeta=1$,
$\log\eta_h$ is a zero-drift finite-variance random walk and
oscillates almost surely \citep{chungfuchs1951}.
\end{proof}

\begin{corollary}[Summability and finite exceedances]
\label{cor:summability}

If $0<\zeta<1$, then for every $r>0$,
\begin{equation}
\sum_{h=1}^{\infty}\eta_h^r<\infty
\quad\text{a.s.},
\qquad
N_\varepsilon
=
\#\{h:\eta_h>\varepsilon\}<\infty
\quad\text{a.s.}
\label{eq:summability}
\end{equation}

\end{corollary}

\begin{proof}
The root test follows from
$\eta_h^{1/(h-1)}\to\zeta<1$; finite threshold exceedances are then
immediate.
\end{proof}

For $r=2$, this applies directly to the column-specific variance
multiplier in \eqref{eq:loading}. The result is almost sure rather
than integrable: $\E(\eta_h^2)=\infty$ for $h\geq2$, so
$\E(\sum_h\eta_h^2)=\infty$ despite almost-sure summability.

\begin{theorem}[Central limit theorem]
\label{thm:clt}

As $h\to\infty$,
\begin{equation}
\frac{
\log\eta_h-(h-1)\log\zeta
}{
(\pi/2)\sqrt{h-1}
}
\xrightarrow{d}
N(0,1).
\label{eq:clt}
\end{equation}

\end{theorem}

\begin{proof}
Apply the classical central limit theorem to the sum in
\eqref{eq:decomp}.
\end{proof}

\section{Probability calibration and effective dimension}

Let
\[
Z_H
=
\sum_{\ell=2}^{H}Y_\ell,
\qquad
H\geq2.
\]
Its characteristic function is
\[
\phi_{Z_H}(t)
=
\sech^{H-1}\left(
\frac{\pi t}{2}
\right).
\]

The Gil--Pelaez inversion formula
\citep{gilpelaez1951} gives
\begin{equation}
F_H(z)
=
\frac12
+
\frac{1}{\pi}
\int_0^\infty
\frac{\sin(tz)}{t}
\sech^{H-1}\left(
\frac{\pi t}{2}
\right)dt.
\label{eq:cdf}
\end{equation}

Let $q_{p,H}=F_H^{-1}(p)$.

\begin{proposition}[Probability-based calibration]
\label{prop:calibration}

For $H\geq2$,
\begin{equation}
Q_p(\eta_H)
=
\zeta^{H-1}
e^{q_{p,H}},
\qquad
\Pp(\eta_H\leq\varepsilon)=p
\Longrightarrow
\boxed{
\zeta
=
\exp\left\{
\frac{
\log\varepsilon-q_{p,H}
}{
H-1
}
\right\}.
}
\label{eq:calibration}
\end{equation}

\end{proposition}

\begin{proof}
Since
$\eta_H=\zeta^{H-1}\exp(Z_H)$,
monotonicity of the exponential gives the quantile expression;
solving $Q_p(\eta_H)=\varepsilon$ gives
\eqref{eq:calibration}.
\end{proof}

The calibrated value satisfies $\zeta<1$ if and only if
$\log\varepsilon<q_{p,H}$.

Symmetry implies
\[
q_{1-p,H}=-q_{p,H},
\qquad
q_{0.5,H}=0,
\]
recovering median calibration
\begin{equation}
\zeta
=
\varepsilon^{1/(H-1)}.
\label{eq:medianCalibration}
\end{equation}

For fixed $h$, $\eta_h$ is increasing in $\zeta$ under a
common-increment coupling. Therefore
$\Pp(\eta_H\leq\varepsilon)$ is decreasing in $\zeta$, guaranteeing
uniqueness of the calibration.

Equation \eqref{eq:calibration} is a marginal statement about the
specific column $H$; it does not imply that all subsequent columns
remain below $\varepsilon$. In the motivating model,
$\Pp(\eta_H\leq\varepsilon)=p$ means that the MHCP contribution to
the column-$H$ standard deviation is at most $\varepsilon$ with
probability $p$.

For example,
\[
H=10,\qquad
\varepsilon=0.05,\qquad
p=0.90
\]
gives
\begin{equation}
q_{0.90,10}\approx5.967,
\qquad
\zeta\approx0.369.
\label{eq:workedCalibration}
\end{equation}
Median calibration instead gives
$\zeta\approx0.717$.

A Monte Carlo check using $2\times10^5$ draws gave empirical probabilities
$0.5001$, $0.7483$, and $0.9012$ for target probabilities
$0.50$, $0.75$, and $0.90$, respectively.

For large $H$, Theorem~\ref{thm:clt} gives
\begin{equation}
q_{p,H}
\approx
\frac{\pi}{2}
\sqrt{H-1}\,
\Phi^{-1}(p).
\label{eq:qapprox}
\end{equation}
For $(\varepsilon,p)=(0.05,0.90)$, the semi-analytic and
normal-approximate values of $\zeta$ are respectively
$(0.1773,0.1728)$, $(0.3694,0.3665)$, and
$(0.5397,0.5382)$ at $H=5,10,$ and $20$.

For a global threshold-based summary, write
\[
N_\varepsilon
=
\sum_{h=1}^{\infty}
\mathbf{1}\{\eta_h>\varepsilon\}.
\]
Tonelli's theorem gives
\begin{equation}
\E(N_\varepsilon)
=
\mathbf{1}\{\varepsilon<1\}
+
\sum_{h=2}^{\infty}
\left[
1-
F_h\{
\log\varepsilon-(h-1)\log\zeta
\}
\right].
\label{eq:expectedN}
\end{equation}

For $0<\zeta<1$, this expectation is finite. Indeed, for sufficiently
small $q\in(0,1)$,
\[
\rho_q
=
\zeta^q
\sec\left(
\frac{\pi q}{2}
\right)
<1,
\]
and Markov's inequality together with \eqref{eq:fractional-moment} gives
\[
\Pp(\eta_h>\varepsilon)
\leq
\varepsilon^{-q}
\rho_q^{\,h-1},
\]
which is geometrically summable.

Thus
$\Pp(\eta_H\leq\varepsilon)=p$
and $N_\varepsilon$ answer complementary questions: the former
controls one chosen column, whereas the latter summarizes overall
threshold-based effective dimension.

Numerical evaluation of \eqref{eq:expectedN} gives
$\E(N_{0.05})\approx4.75$ for the worked calibration, whereas
median calibration gives approximately $20.63$.

\begin{table}[t]
\centering
\caption{
Effect of the target probability on MHCP calibration and
threshold-based effective dimension for
$H=10$ and $\varepsilon=0.05$.
}
\label{tab:design}
\begin{tabular}{ccc}
\hline
$p$ & Calibrated $\zeta$ & $\E(N_{0.05})$ \\
\hline
0.50 & 0.7169 & 20.63 \\
0.75 & 0.5077 & 7.60 \\
0.90 & 0.3694 & 4.75 \\
\hline
\end{tabular}
\end{table}

Increasing the target probability therefore induces more aggressive
shrinkage and a substantially smaller threshold-based effective
dimension.

\section{Relation to existing shrinkage processes}

For a positive multiplicative process
$\Lambda_h=\prod_{\ell=1}^{h}D_\ell$ with i.i.d. increments
satisfying $\E|\log D_1|<\infty$,
\[
h^{-1}\log\Lambda_h
\longrightarrow
\E\{\log D_1\}
\qquad\text{a.s.}
\]
The multiplicative gamma process
\citep{bhattacharya2011,durante2017}
uses gamma increments, whereas cumulative shrinkage uses an ordered
mixture mechanism \citep{legramanti2020}. The MHCP has the
particularly transparent representation \eqref{eq:decomp}, in which
$\zeta$ changes the drift without changing the centered increment
law. These are structural prior comparisons rather than claims of
universal posterior superiority.

\section{Illustration and discussion}

\begin{figure}[!t]
\centering
\includegraphics[width=0.94\textwidth]{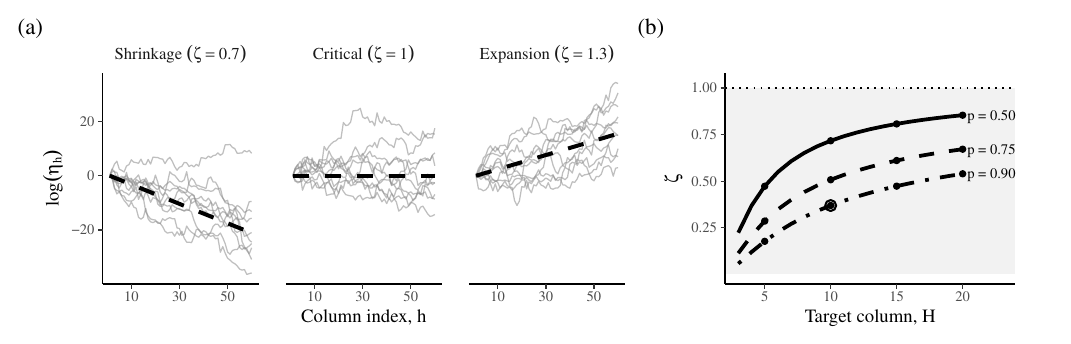}
\caption{MHCP pathwise behavior and calibration.
(a) Ten simulated log-scale trajectories in shrinking ($\zeta=0.7$),
critical ($\zeta=1$), and expanding ($\zeta=1.3$) regimes, with theoretical
median trends; simulation seed 2026.
(b) Values of $\zeta$ satisfying $P(\eta_H\leq0.05)=p$ for
$p=0.50$, $0.75$, and $0.90$. At $H=10$, the corresponding values are
$0.7169$, $0.5077$, and $0.3694$. The dotted line marks $\zeta=1$; the
open circle denotes $(H,p)=(10,0.90)$.}
\label{fig:main}
\end{figure}

Figure~\ref{fig:main} summarizes the main interpretation.
Individual paths can exhibit substantial local variation while their
long-run behavior is governed by $\zeta$, and interpretable
column-specific probability statements can be converted directly
into prior choices.
This combination of almost-sure summability with infinite marginal
moments illustrates how the MHCP can enforce strong pathwise
regularization without eliminating heavy-tailed prior protection
against occasional large components.

For $\zeta<1$, an infinite MHCP sequence is almost surely effectively
finite at every fixed positive threshold while retaining explicit
probabilities of local rebounds. The marginal calibration
$\Pp(\eta_H\leq\varepsilon)=p$ and the global summary
$N_\varepsilon$ therefore provide complementary views of shrinkage
and effective dimension.

Although motivated by canonical correlation analysis, the same
ordered-scale mechanism applies to Bayesian factor, matrix, tensor,
and multi-view latent representations. An accompanying R package,
\texttt{mhcp}, provides simulation, probability calibration,
threshold-exceedance summaries, and reproducible figure code. Version 0.1.0 is archived on Zenodo (DOI: 10.5281/zenodo.23355811) and development is available on GitHub:
\url{https://github.com/SiddheshKulkarni-source/mhcp}.

\textbf{Declaration of competing interest}: The author declares no conflict of interest.



\end{document}